\documentclass{article}

\usepackage{arxiv}

\usepackage{amsfonts}
\usepackage{amsmath}
\usepackage{subfig}
\usepackage{graphicx}
\usepackage{amssymb}
\usepackage{mathtools}
\usepackage{xcolor}
\usepackage{caption}
\usepackage{float}
\usepackage{algorithm}
\usepackage{algpseudocode}
\usepackage{stmaryrd}
\usepackage{epstopdf}
\usepackage{multirow}
\usepackage{color,soul}
\usepackage{accents}
\usepackage{float}
\usepackage{verbatim}
\usepackage{url}

\newcommand{\R}{{\mathbb{R}}}

\newcommand{\PJ}[1]{\textcolor{red}{PJ: #1}}

\newtheorem{theorem}{Theorem}[section]

\newtheorem{assumption}{Assumption}

\newtheorem{definition}[theorem]{Definition}

\newtheorem{remark}[theorem]{Remark}
\newtheorem{problem}[theorem]{Problem}
\newtheorem{proof}[theorem]{Proof}

\title{Reach-Avoid-Stay-Collision-Avoidance Negotiation Framework for Multi-Agent Systems via Spatiotemporal Tubes
\thanks{This work was supported in parts by the ARTPARK, the Department of Science and Technology, the Government of Karnataka, and the Siemens. The work of Ratnangshu Das was supported by the Prime Minister’s Research Fellowship from the Ministry of Education, Government of India.}
}

\author{
 Mohd. Faizuddin Faruqui \\
  Centre for Cyber-Physical Systems, IISc, Bengaluru, India \\
  U.R. Rao Satellite Centre, ISRO, Bengaluru, India \\
  \texttt{mfaruqui@ursc.gov.in} \\
   \And
 Ratnangshu Das \\
  Centre for Cyber-Physical Systems\\
  IISc, Bengaluru, India\\
  \texttt{ratnangshud@iisc.ac.in} \\
  \And
 Ravi Kumar L \\
  U.R. Rao Satellite Centre, 
  ISRO, Bengaluru, India\\
  \texttt{rkkumarl@ursc.gov.in} \\
  \And
 Pushpak Jagtap \\
  Centre for Cyber-Physical Systems, and \\
  Department of Aerospace Engineering, \\
  IISc, Bengaluru, India\\
  \texttt{pushpak@iisc.ac.in} \\
}

\begin{document}
\maketitle

\begin{abstract}
This study presents a multi-agent negotiation-based framework to obtain collision-free paths while performing prescribed-time reach-avoid-stay (RAS) tasks for agents with unknown dynamics and bounded disturbance. By employing spatiotemporal tubes to generate time-varying state constraints, we ensure that all agents adhere to RAS specifications using synthesized controllers. To prevent inter-agent collisions, a negotiation mechanism is proposed where successful negotiations result in spatiotemporal tubes for each agent fulfilling desired tasks. This approach results in a completely distributed, approximation-free control law for each agent. The effectiveness of this mechanism was validated through simulations of multi-agent robot navigation and drone navigation tasks involving prescribed-time RAS specifications and collision avoidance. 
\end{abstract}

\section{Introduction}
Multi-agent systems offer several advantages over their monolithic counterparts, particularly in high-stake and safety-critical arenas. Several applications in multi-agent systems, such as consensus and connectivity maintenance, have attracted further interest in complex scenarios such as obstacle and collision avoidance \cite{olfati-saber_flocking_2006}, \cite{sun2017robust}. These tasks are classified as safety tasks in \textit{formal methods} parlance, which is the primary requirement of any autonomous system in path and motion planning scenarios. Reach-Avoid-Stay (RAS) tasks enable agents to reach and stay in a target set from an initial set while avoiding unsafe sets. Controllers synthesized for RAS specifications enforce state constraints on the agent and ensure that the state trajectory meets RAS specifications. Additionally, for multi-agent autonomous systems, inter-agent collision avoidance must also be ensured.  
\par As safety enforcement on continuous time dynamical systems will require constraining temporal behavior, feedback controllers manipulating control input on-the-go may not provide formal guarantees. Alternatively, synthesis methods such as symbolic control with formal guarantees have gained momentum as a more robust approach. Symbolic control methods involve the synthesis of controllers that satisfy safety specifications for a sound abstraction of the original system \cite{hutchison_approximately_2012}, \cite{sundarsingh2023scalable}. Similarly, for an abstracted multi-agent system, the finite states can be represented as vertices of a graph, on which strategies for satisfying specifications that are expressed as Linear Temporal Logic (LTL) can be formulated \cite{fijalkow2023games}, \cite{samuel2023efficientcontrollersynthesistechniques}. 
\par However, studies related to controller synthesis for RAS specifications have focused on different approaches to arrive at correct-by-construction controllers for RAS tasks \cite{Das2023Funnel-based}, \cite{9867253}, \cite{9390193}. Path planning for multi-agent scenarios where the arena is compact and high-stakes cannot rely on dynamic feedback strategies that may result in deadlock or stalemate situations \cite{Mylvaganam2017A}. Studies related to multi-agent collision and obstacle avoidance have tackled the problem using controllers based on dynamic inter-agent feedback which is dependent on inter-agent communication despite being distributed \cite{Mylvaganam2017A}, \cite{DAI20172068}. Real-time communication between agents was used in negotiation for collision avoidance in \cite{article} for dynamic adaptability. The prime hurdles identified in addressing RAS with Collision Avoidance (RASCA)  tasks for multi-agent systems were requirement of real-time reliable communication among agents and knowledge of agent dynamics. We consider systems with unknown dynamics in this paper. However, for unknown system dynamics, reliance on real-time communication and dynamic feedback can compromise the safety-critical specifications as well as prescribed time guarantees.

To cope with these issues, we propose a distributed abstraction-free approach that works directly with the system states in continuous time, enforcing temporal state constraints to satisfy the RASCA specifications for unknown dynamics. The main contribution of this study is a negotiation agreement strategy that provides a set of parameterized spatiotemporal tubes guaranteeing collision-free corridors to the agents. Negotiation as a tool for conflict resolution is effectively used in multi-agent scenarios involving coalitions \cite{1668250}. The negotiation object among the agents in this study is a temporally valid tube such that the collision avoidance specification is satisfied. In this study, once the negotiation terminates, all agents are assigned a set of parameterized tubes and subsequently, the synthesized controllers. The network topology adopted in this study facilitates distributed navigation, as once the controllers are synthesized, agents are not dependent on communication.


\section{Notation and Preliminaries}
\label{Notation}
The symbols $\mathbb{R}, \mathbb{R}_0^+, \mathbb{R}^+, \mathbb{Z}^+$ and $\mathbb{Z}_0^+$ denote sets of reals, non-negative reals, positive reals, positive integers and non-negative integers respectively. A column vector with $n$ rows is represented as $\mathbb{R}^n$ and the vector space with $m$ rows and $n$ columns is represented as $\mathbb{R}^{m \times n}$, where $n,m \in \mathbb{Z}^+$. Given an index set $\mathcal{I} = \lbrace 1,2,3,...,N \rbrace$ of $N \in \mathbb{Z}^+$ agents and sets $X_i$, $i \in \mathcal{I}$, the Cartesian product of these sets is given by $\prod_{i\in \mathcal{I}} X_i:= X_1 \times X_2 \times .... \times X_N = \lbrace (x_1,x_2,...,x_N)|x_i \in X_i, i\in \mathcal{I}\rbrace$. The empty set is denoted by $\emptyset$.
\subsection{System and Topology}
\label{SecIIa}
A multi-agent system consisting of $N \in \mathbb{Z}^+$ agents can be expressed as $\Sigma = \lbrace \Sigma_1,\Sigma_1,...,\Sigma_N \rbrace$. The dynamics of the $i^{th}$ agent $\Sigma_i$, $i\in\mathcal{I}=\{1,2,\ldots,N\}$, considered as a control affine system is expressed as: \\
\begin{equation}
\Sigma_i: \dot{x}_i(t) = f_i(x_i(t)) + g_i(x_i(t))u_i(t) + d_i(t),
\label{SysDyn}
\end{equation}
where $x_i(t) \in X_i \subseteq \mathbb{R}^{n_i}, u_i(t) \in \mathbb{R}^{m_i}$ are the state and input vectors of the agent, respectively, and $d_i(t)$ is an unknown bounded disturbance. The functions $f_i: \mathbb{R}^{n_i} \rightarrow \mathbb{R}^{n_i}$ and $g_i: \mathbb{R}^{n_i} \rightarrow \mathbb{R}^{n_i \times m_i}$ are called the flow drift and input matrix, respectively.
\begin{assumption}
\label{assum0}
    The functions $f_i$ and $g_i$ for all the agents $i\in\mathcal{I}$ are locally Lipschitz and $\emph{unknown}.$ 
\end{assumption}
\begin{remark}
    The proposed control framework for control-affine systems can be naturally extended to general nonlinear dynamics by treating the unmodeled nonlinear terms as an additional disturbance. This makes the approach applicable to a broader class of nonlinear systems without requiring explicit system identification.
\end{remark}
\par The state spaces $X_i$ of each agent in the multi-agent system $\Sigma$ are closed and connected. State space X of the multi-agent system $\Sigma$ is the Cartesian product of $N$ state spaces $X_i \subseteq \mathbb{R}^{n_i}, i \in \mathcal{I}$ and expressed as: $\prod_{i\in \mathcal{I}} X_i:= X_1 \times X_2 \times .... \times X_N$.

\par Graphs were used to establish the network topology between agents for communication during negotiation. Consider an undirected graph $G = (\mathcal{I},\mathcal{E})$, such that $\mathcal{E} \subseteq \mathcal{I} \times \mathcal{I}$ is the set of edges that represents bidirectional communication between agents represented by vertices, where $\mathcal{I} = \lbrace 1,2,3, ..., N \rbrace $ is the set of $N$ agents \cite{Mesbahi2010GraphTM}. In this work, we considered a fully connected network in which every agent can communicate with every other agent until the negotiation terminates. We use the notation $\mathcal{N}_i$ to represent the set of communicating neighborhood agents of agent $i$ defined as $\mathcal{N}_i = \mathcal{I} \symbol{92} \lbrace i \rbrace$. We assume that in this fully connected network of agents, the $i^{th}$ agent possessing the token, functions as the hub of the star topology, analogous to computer networks \cite{tanenbaum2003computer}. The neighboring agents update their path plan to the hub agent through spatiotemporal tubes (STT) \cite{das2024prescribed}, as explained in later sections. The order of token passing can be either random or prioritized agent-wise, and we consider the sequential order in this study. 
\subsection{Problem Formulation}
Reach-avoid-stay (RAS) tasks are a classical robot navigation problem that occurs under various scenarios in real-life robotic applications. The problem statement in this study extends this work to multi-agent systems with an additional specification of collision avoidance (CA). The prescribed time RAS with collision avoidance (RASCA) task addressed in this study can be defined as in Definition \ref{Def2.1}.

\begin{definition}[Prescribed time RASCA task]
\label{PT-RASCA-Def}
Given a set of agents indexed in $\mathcal{I}$ operating in the state space X, the communication graph among the agents $G = (\mathcal{I},\mathcal{E})$ and the unsafe set $U_i \subset X_i$, let $S_i \subseteq X_i \setminus U_i$ be the initial set, $T_i \subseteq X_i \setminus U_i$ be the target set and $t_p^{(i)} \in \R^+$ be the prescribed time to complete the task for each agent $i$, where $i \in \mathcal{I}$.
We say that each agent $i$ satisfies the prescribed-time RAS task if for the given initial position $x_i(0) \in S_i, \: \exists t \in [0, t_p^{(i)}]$ such that $x_i(t) \in T_i$ and for all $s \in [0,t_p^{(i)}]$, $x_i(s) \in X_i \setminus U_i$. 
Additionally, the agent $i$ satisfies the collision avoidance (CA) task if for all $t \in [0, t_p^{(i)}], x_i(t) \neq x_j(t), \forall j \in \mathcal{N}_i$.
\label{Def2.1}
\end{definition}
\par To satisfy the prescribed-time RAS task for each individual agent, this study leverages the STT framework \cite{das2024prescribed} and invokes the following assumption:
\begin{assumption}
    The static obstacles $O_j^i \subseteq U_i \subset X_i, j \in [1;n^{(u)}_i]$, where $n^{(u)}_i$ is the number of static obstacles, initial sets $S_i$ and target sets $T_i$ are assumed to be convex and compact. It is assumed that the agents' initial location and target location are both obstacle-free (i.e.,
    $\forall i \in \mathcal{I}, S_i \cap O^i_j = T_i \cap O^i_j = \emptyset, \forall j \in [1;n^{(u)}_i]$) and collision-free (i.e., 
    $\forall i \in \mathcal{I}, S_i \cap S_j = T_i \cap T_j = \emptyset, \forall j \in \mathcal{N}_i$).
    \label{assum1}
\end{assumption}
\begin{remark}
    The unsafe set $U_i = \bigcup_{j\in [1;n^{(u)}_i]} O_j^i,$ can be concave and disconnected.
    Further, the agent arena can be over-approximated into a hyper-rectangle: $\llbracket X_i \rrbracket=\prod_{k=[1;n_i]}\left[\underline{X}_{i}^{k}, \overline{X}_i^{k}\right], \text{ where } \underline{X}_{i}^{k} := \min \left\{x_i^{(k)} \in \mathbb{R} \mid[x_i^{(1)},\ldots, x_i^{(n_i)}] \in X_i\right\}, \overline{X}_i^{k}:=\max \{x_i^{(k)} $ $ \in \mathbb{R} \mid\left[x_i^{(1)}, \ldots, x_i^{(n_i)}\right] \in X_i\}$ are the projections of the space $X_i$ on the $k^{th}$ dimension and $[x_i^{(1)},\ldots,x_i^{(n_i)}]$ is the state vector of agent i in state space $X_i$, and the unsafe set is expanded as $\llbracket U_i \rrbracket = U_i \cup \llbracket X_i \rrbracket \setminus X_i$. Without loss of generality, in the remainder of the paper, we will assume $X_i$ to be a hyper-rectangle unless stated otherwise.
\end{remark}
Thus, we formally present the problem statement addressed in this paper.
\begin{problem}
       Given the multi-agent system $\Sigma$, the state space X, communication graph $G = \lbrace \mathcal{I}, \mathcal{E} \rbrace$, set of initial sets $S=\lbrace S_i \rbrace_{i\in \mathcal{I}}$ and set of target sets $T=\lbrace T_i \rbrace_{i\in \mathcal{I}}$, such that all agents belonging to the multi-agent system $\Sigma$ satisfy Assumption \ref{assum0} and \ref{assum1},  design a negotiation mechanism that achieves collision avoidance to ensure the satisfaction of the prescribed time RASCA task as in Definition \ref{Def2.1}. 
\end{problem}

\section{Spatiotemporal Tubes (STT)}
\label{STT}
To enforce the Prescribed-Time RASCA (PT-RASCA) specifications for each agent, we leverage the STT framework \cite{das2024prescribed}.

\begin{definition}[Spatiotemporal Tubes for PT-RASCA]
\label{def} 
Given the prescribed-time RASCA task as in Definition \ref{PT-RASCA-Def}, an STT for an agent $i$ is defined as a time-varying interval {$\mathcal{R}_i(t) = \prod_{k=1}^{n_i}\mathcal{R}_i^{(k)}(t)$, where, $\mathcal{R}_i^{(k)}(t) = [\gamma^{(k)}_{i,L}(t), \gamma^{(k)}_{i,U}(t)]$, $k \in [1;n_i]$}, with $\gamma_{i,L}^{(k)}(t):\R_0^+ \rightarrow \R$ and $\gamma_{i,U}^{(k)}(t):\R_0^+ \rightarrow \R$ are continuously differentiable functions satisfying $\gamma_{i,L}^{(k)}(t) < \gamma_{i,U}^{(k)}(t)$ for all times $t \in [0, t_p^{(i)}]$ and for each dimension $k \in [1;n_i]$, with the following properties: 
\begin{align}\label{eqn:stt}
    &\prod_{k=1}^{n_i} [\gamma_{i,L}^{(k)}(t), \gamma_{i,U}^{(k)}(t)] \subseteq X_i, \forall t \in [0,t_p^{(i)}], \nonumber \\
    &\prod_{k=1}^{n_i} [\gamma_{i,L}^{(k)}(0), \gamma_{i,U}^{(k)}(0)] \subseteq S_i, \
    \prod_{k=1}^{n_i} [\gamma_{i,L}^{(k)}(t_p^{(i)}), \gamma_{i,U}^{(k)}(t_p^{(i)})] \subseteq T_i, \nonumber \\
    &\prod_{k=1}^{n_i} 
    [\gamma_{i,L}^{(k)}(t), \gamma_{i,U}^{(k)}(t)] \cap U_i = \emptyset, \forall t \in [0,t_p^{(i)}].
\end{align}
\end{definition}
As discussed in \cite{das2024prescribed}, the generation of the STT involves three primary steps: (i) designing reachability tubes $\mathcal{R}_i(t)$ to guide the system trajectory from the initial set $S_i$ to the target set $T_i, \forall i\in \mathcal{I}$, (ii) introducing a circumvent function to avoid the unsafe set $U_i$, and (iii) adjusting the tubes around the unsafe region through an adaptive framework to obtain obstacle-circumvented tubes
$\tilde{\mathcal{R}}_i(t) = \prod_{k=1}^{n_i}\tilde{\mathcal{R}}_i^{(k)}(t)$, where, $\tilde{\mathcal{R}}_i^{(k)}(t) = [\tilde{\gamma}^{(k)}_{i,L}(t),\tilde{\gamma}^{(k)}_{i,U}(t)]$.


However, in this study, to handle the collision avoidance (CA) task, after step (iii), we introduce an additional negotiation step that provides us with the parameterized spatiotemporal tubes $\hat{\mathcal{R}}_i^{(k)}(t) = [\hat{\gamma}^{(k)}_{i,L}(t),\hat{\gamma}^{(k)}_{i,U}(t)], \forall i \in \mathcal{I}, k \in [1;n_i]$ as explained in Section \ref{Negotiation}. Let $\mathcal{N}_i^c \subseteq \mathcal{N}_i$ be the subset of neighbors of agent $i$ with which its tube collide, i.e. satisfy (3) and thereby includes candidate agents with which the agent $i$ will negotiate, let ${\mathcal{R}}(t) = \lbrace \mathcal{R}_i(t) \rbrace _{i \in \mathcal{I}}$ be the set of Spatiotemporal tubes and $\hat{\mathcal{R}}(t) = \lbrace \hat{\mathcal{R}}_i(t)  \rbrace _{i \in \mathcal{I}}$ be the set of all negotiated tubes.  Collision for an agent $i$ can happen when its Spatiotemporal tube intersects with other agents' tube, i.e., 
\begin{align}
\label{Coll_Cond}
&\exists t \in [0, \min (t_p^{(i)}, t_p^{(j)})], \text{ such that,}\nonumber \\
&\quad \mathcal{R}_i(t) \cap {\mathcal{R}}_j(t) \neq \emptyset, \forall i \in \mathcal{I},  j \in \mathcal{N}_i^c.
\end{align}

{\section{Negotiation using parameterized tubes}\label{Negotiation}}
Providing collision avoidance guarantees for reactive agents of a multi-agent system in a restrictive arena is a complex task \cite{gelbal2019cooperative}. 
Although the tubes can be circumvented in the time intervals of collision to avoid the collision condition mentioned in (\ref{Coll_Cond}), it may lead to intersections with other tubes. To mitigate this problem, we propose using parameterized spatiotemporal tubes. 
The collision interval and parameterized tubes are defined below.

\begin{definition}[Collision Interval for an Agent]
\label{Def4.1}
 Given the multi-agent system $\Sigma$ satisfying Assumption \ref{assum0} and \ref{assum1} and tubes $\lbrace \mathcal{R}_i(t) \rbrace _{i\in \mathcal{I}}$, the collision interval for an agent $i \in \mathcal{I}$ is defined by the closed time interval $[\underaccent{\bar}{t}^{(i)},\Bar{t}^{(i)}]$, where 
\begin{equation}
\begin{aligned}
        \underline{t}^{(i)} = \min_{j \in \mathcal{N}_i^c}(t \in [0, t^{(i)}_p] \ | \ {\mathcal{R}}_i(t) \cap {\mathcal{R}}_j(t) \neq \emptyset),\\
     \overline{t}^{(i)} = \max_{j \in \mathcal{N}_i^c}(t \in [0, t^{(i)}_p] \ | \ {\mathcal{R}}_i(t) \cap {\mathcal{R}}_j(t) \neq \emptyset).
     \end{aligned}
\end{equation}
\end{definition}

The start of this interval, $\underline{t}^{(i)}$, is defined as the earliest time at which a collision with any neighboring agent $j \in \mathcal{N}_i^c$ occurs, while $\overline{t}^{(i)}$ is the latest time of collision. This interval, $[\underaccent{\bar}{t}^{(i)},\Bar{t}^{(i)}]$, defines the duration during which agent $i$’s Spatiotemporal tube intersects with that of another agent and indicates the time window within which an adjustment is needed.

\begin{definition}
\label{TubePara}
(Parameterized Spatiotemporal Tube): Given an agent $i\in \mathcal{I}$ in a multi-agent system $\Sigma$ satisfying Assumption \ref{assum0}, with the collision interval $[\underline{t}^{(i)}, \overline{t}^{(i)}]$, the temporally parameterized tube is defined as:
\begin{equation}
    \hat{\mathcal{R}}_i(t) = \begin{cases} 
      {\mathcal{R}}_i(t), &\forall t \in [0, \underline{t}^{(i)}), \\
      {\mathcal{R}}_i(\underline{t}^{(i)} - \delta), &\forall t \in [\underline{t}^{(i)}, \overline{t}^{(i)}], \\
      {\mathcal{R}}'_i(t-\overline{t}^{(i)}), &\forall t \in (\overline{t}^{(i)}, t_p^{(i)}], 
   \end{cases}
\end{equation}
\end{definition}
where ${\mathcal{R}}'_i(t-\overline{t}^{(i)})$ is a tube with its initial set redefined as the hyper-polygon $\mathcal{R}_i(\underline{t}^{(i)}-\delta)$, $\delta$ is the factor which keeps the tube frozen in safe space outside the collision zone and the updated time allotted to complete the task is now $t_p^{(i)}-\overline{t}^{(i)}$.

Thus, the parameterized tubes are structured in three segments. During the pre-collision period $t \in [0, \underline{t}^{(i)})$, the parameterized tubes are set equal to the spatiotemporal tubes, allowing the agent to follow its original path. In the collision interval $t \in [\underline{t}^{(i)}, \overline{t}^{(i)}]$, the agent’s tube is spatially “frozen” before the start of the collision zone at $\mathcal{R}(\underline{t}^{(i)}-\delta)$, preventing it from entering into the collision area. Once the collision interval is over, for $t \in (\overline{t}^{(i)}, t_p^{(i)}]$, the tubes are redefined to guide the agent toward the target set from its current position within the remaining time $t_p^{(i)}-\overline{t}^{(i)}$. This redefinition allows the agent to complete its task within the adjusted timeline.



The NEGOTIATECOLLISION algorithm in Algorithm \ref{algo-1} operates iteratively to resolve collisions within a network of agents by adjusting their assigned Spatiotemporal tubes.   
\begin{algorithm}[]
    \caption{NEGOTIATECOLLISION(G, ${{\mathcal{R}}}(t)$)}
    \textbf{Input:} Spatiotemporal Tubes: ${\mathcal{R}}_i(t), \forall i \in \mathcal{I}$ \\
    \textbf{Output:} Collision-free parameterized tubes: $\hat{\mathcal{R}}_i(t), \forall i \in \mathcal{I}$ 
    \begin{algorithmic}[1]
        \State Initialize: $\hat{\mathcal{R}}_i(t) \leftarrow {\mathcal{R}}_i(t), \forall t \in [0,t_p^{(i)}], \forall i \in \mathcal{I}$
            \State \textbf{Set} iter $\leftarrow 1$
            \State \textbf{Set} updated $\leftarrow$ True
        \While{updated}
            \State \textbf{Set} updated $\leftarrow$ False
            \For{$i \in \mathcal{I}$}
            \If{$\exists (t,j) \in [0, \min (t_p^{(i)}, t_p^{(j)})] \times \mathcal{N}_i,$ such that, ${\mathcal{R}}_i(t) \cap {\mathcal{R}}_j(t) \neq \emptyset$}
                \State \texttt{Compute Collision Interval:} 
                $$\underline{t}^{(i)} \leftarrow \min_{j \in \mathcal{N}_i^c}(t \in [0, t^{(i)}_p] \ | \ {\mathcal{R}}_i(t) \cap {\mathcal{R}}_j(t) \neq \emptyset),$$
                $$\overline{t}^{(i)} \leftarrow \max_{j \in \mathcal{N}_i^c}(t \in [0, t^{(i)}_p] \ | \ {\mathcal{R}}_i(t) \cap {\mathcal{R}}_j(t) \neq \emptyset).$$
                \State \texttt{Update parameterized tubes:}
                $$\hat{\mathcal{R}}_i(t) \leftarrow \begin{cases} 
                  {\mathcal{R}}_i(t), &\forall t \in [0, \underline{t}^{(i)}), \\
                  {\mathcal{R}}_i(\underline{t}^{(i)}-\delta), &\forall t \in [\underline{t}^{(i)}, \overline{t}^{(i)}], \\
                  {\mathcal{R}}'_i(t-\overline{t}^{(i)}), &\forall t \in (\overline{t}^{(i)}, t_p^{(i)}], 
               \end{cases}$$
               \State \textbf{Set} updated $\leftarrow$ True
               \EndIf
            \EndFor
            \State \textbf{Set} iter $\leftarrow$ iter + $1$
            \State $\mathcal{R}_i(t) \leftarrow \hat{\mathcal{R}}_i(t), \forall t \in [0,t_p^{(i)}], \forall i \in \mathcal{I}$
        \EndWhile 
        \State \textbf{return} $\hat{\mathcal{R}}_i(t),  \forall i \in \mathcal{I}$
    \end{algorithmic}
\label{algo-1}
\end{algorithm}

The algorithm iteratively modifies the tubes by introducing temporal halts to ensure that all agents navigate without colliding. 
Initially, each agent's parameterized tube $\hat{\mathcal{R}}_i$ is set equal to their original Spatiotemporal tubes ${\mathcal{R}}_i$. 
The algorithm then enters a loop that checks for intersections between the Spatiotemporal tubes of each agent $i$ and its neighboring agents $j$ over the time interval $[0, t_p^{(i)}]$. 
If a collision is detected, it identifies the collision interval $[\underline{t}^{(i)}, \overline{t}^{(i)}]$, during which agent $i$’s tube is “frozen” spatially at the entry point of the collision zone to prevent it from entering it. 
After the interval ends, the tube is redefined to guide the agent to its target within the remaining time. This procedure is repeated for each agent, iteratively adjusting their tubes until all are collision-free.
The algorithm concludes by returning the modified, collision-free parameterized tubes, enabling coordinated and safe navigation for all agents. 
Through this process of detecting collision intervals and introducing delays, the algorithm ensures collision-free, efficient navigation for each agent.


\begin{theorem}
    Given the multi-agent system $\Sigma$ satisfying Assumption \ref{assum0} and \ref{assum1}, the state space X, communication graph $G = \lbrace \mathcal{I}, \mathcal{E} \rbrace$, set of initial sets $S=\lbrace S_i \rbrace_{i\in \mathcal{I}}$ and set of target sets $T=\lbrace T_i \rbrace_{i\in \mathcal{I}}$, the parameterized tubes $\hat{\mathcal{R}}_i(t)$ obtained from the NEGOTIATECOLLISION algorithm satisfy the following:
    $$\hat{\mathcal{R}}_i(t) \cap \hat{\mathcal{R}}_j(t) = \emptyset, \forall (i,j,t) \in \mathcal{I} \times \mathcal{N}_i^c \times [0,t_p^{(i)}],$$
    which ensures that the parameterized tubes are collision-free.
\end{theorem}
\begin{proof}
    We will prove this by contradiction. Let us assume that the algorithm has terminated and we still have the following, $\hat{\mathcal{R}}_i(t) \cap \hat{\mathcal{R}}_j(t) \neq \emptyset$, for some $i \in \mathcal{I}$, $j\in \mathcal{N}_i^c$ and $t \in [0, {t}_p^{(i)}]$. 
    
    From Algorithm 1, we obtain the collision interval $t \in [\underline{t}^{(i)}, \overline{t}^{(i)}]$. For $t \in [0, \underline{t}^{(i)})$, since $\hat{\mathcal{R}}_i(t)$ and $\hat{\mathcal{R}}_j(t)$ are designed to be disjoint, it follows that 
    $$\forall  t \in [0, \underline{t}^{(i)}), \hat{\mathcal{R}}_i(t) \cap \hat{\mathcal{R}}_j(t) = \emptyset.$$
    For, $t \in [\underline{t}^{(i)}, \overline{t}^{(i)}]$, $\hat{\mathcal{R}}_i(t)$ is spatially frozen outside the collision zone at ${\mathcal{R}}_i(\underline{t}_i-\delta)$. Consequently, 
    $$\forall  t \in [\underline{t}^{(i)}, \overline{t}^{(i)}], \hat{\mathcal{R}}_i(t) \cap \hat{\mathcal{R}}_j(t) = \emptyset.$$ 
    For, $ t \in (\overline{t}^{(i)}, {t_p}^{(i)}]$, if there is no collision, we have 
    $$\forall  t \in (\overline{t}^{(i)}, {t_p}^{(i)}], \hat{\mathcal{R}}_i(t) \cap \hat{\mathcal{R}}_j(t) = \emptyset,$$
    which contradicts our assumption that there exists $t \in [0, {t}_p^{(i)}]$, such that, $\hat{\mathcal{R}}_i(t) \cap \hat{\mathcal{R}}_j(t) \neq \emptyset$. On the other hand,  if a collision reoccurs in this segment, the algorithm identifies a new collision interval $[\underline{t}_{new}^{(i)}, \overline{t}_{new}^{(i)}] \subset (\overline{t}^{(i)}, t_p^{(i)}]$, and the same freezing-and-redefining process repeats, contradicting our assumption on termination of the algorithm.
\end{proof}

In the following Example-1, we explain the iterative structure of the NEGOTIATECOLLISION algorithm, where any potential collision is detected and resolved by freezing the parameterized tubes $\hat{\mathcal{R}}_i(t)$ at safe points, redefining it only after the collision intervals.  \\
\textbf{Example-1:} Consider the set of four agents $\mathcal{I} = \lbrace 1,2,3,4 \rbrace$ of which agents $1$ and $3$ have collision between time zone $t_a$ and $t_b$ and agents $2$ and $4$
have collision between time zone $t_c$ and $t_d$, where $t_a < t_b < t_c < t_d$. The communication topology and sequence of negotiations for the first iteration are shown in Figure \ref{topol}. Agents with collision-free tubes from the start to the target position were not parameterized during the iterations, we consider that the prescribed time for all agents is same as $t_p$. Table \ref{Param_Tube} shows the sequence of negotiations, such that we obtain a collision-free set of tubes. Now, suppose the parameterized tube obtained when agent $1$ mitigates collision with agent $3$ ends up into a new collision with agent $4$ in collision zone $[t_e,t_f]$, then agent $4$ updates its tube to mitigate collision with updated tube of agent $1$. The negotiation iterations continue until all the agents get collision free tubes as shown in Table 
\ref{Param_Tube}. The parameterized tube are then written as in (\ref{Parametric1})-(\ref{Parametric3}).
\begin{figure}[H]
	\centering
	\includegraphics[width=0.5\linewidth]{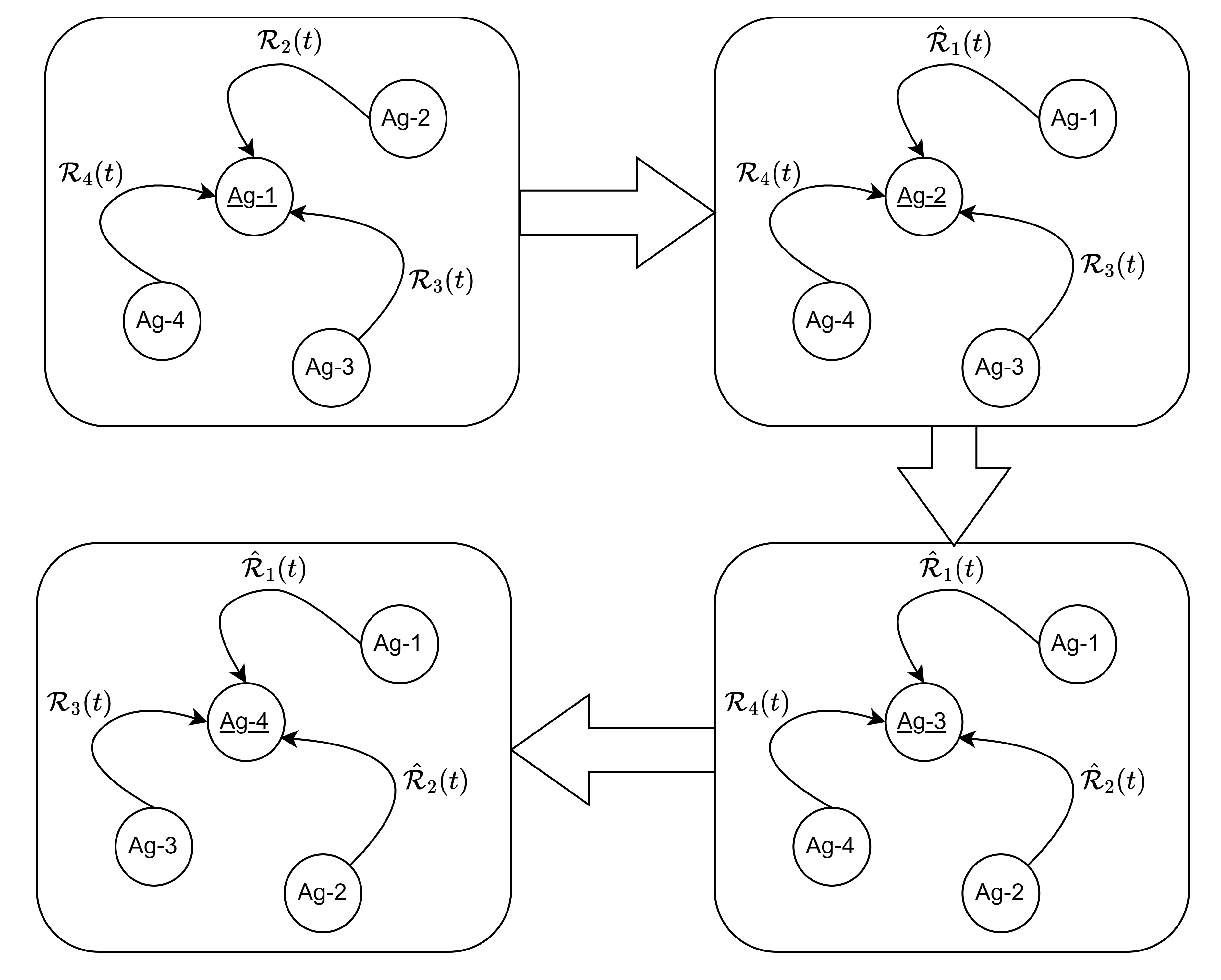}
	\caption{Negotiation Sequence: Iteration-1}
	\label{topol}
\end{figure}

\begin{table}[H]
\centering
	\caption{Tube Parameterization}
	\begin{tabular}{|l|l|l|l|}
	\hline
\multicolumn{1}{|c|}{\textbf{\begin{tabular}[c]{@{}c@{}}Hub\\ Agent\end{tabular}}} &
  \multicolumn{1}{c|}{\textbf{\begin{tabular}[c]{@{}c@{}}Tubes\\ Uploaded\end{tabular}}} &
  \multicolumn{1}{c|}{\textbf{\begin{tabular}[c]{@{}c@{}}Collision\\ Case\end{tabular}}} &
  \multicolumn{1}{c|}{\textbf{\begin{tabular}[c]{@{}c@{}}Updated\\ Tube\end{tabular}}} \\ \hline
       \multicolumn{4}{|c|}{\textbf{Iteration-1}} \\ \hline
 
		1 & ${\mathcal{R}}_2(t),{\mathcal{R}}_3(t),{\mathcal{R}}_4(t)$ & ${\mathcal{R}}_1(t) \cap {\mathcal{R}}_3(t) \neq \emptyset$ & $\hat{{\mathcal{R}}}_1(t)$ \\ \hline
		2 & $\hat{{\mathcal{R}}}_1(t),{\mathcal{R}}_3(t),{\mathcal{R}}_4(t)$ & $\mathcal{R}_2(t) \cap {\mathcal{R}}_4(t) \neq \emptyset$ & $\hat{{\mathcal{R}}}_2(t)$ \\ \hline
		3 & $\hat{{\mathcal{R}}}_1(t),\hat{\mathcal{R}}_2(t),{\mathcal{R}}_4(t)$ & - & ${\mathcal{R}}_3(t)$ \\ \hline
		4 & $\hat{\mathcal{R}}_1(t),\hat{\mathcal{R}}_2(t),{\mathcal{R}}_3(t)$ & ${\mathcal{R}}_4(t) \cap \hat{\mathcal{R}}_1(t) \neq \emptyset$ & $\hat{\mathcal{R}}_4(t)$ \\ \hline
        \multicolumn{4}{|c|}{\textbf{Iteration-2}} \\ \hline
        
		1 & $\hat{\mathcal{R}}_2(t),{\mathcal{R}}_3(t),\hat{\mathcal{R}}_4(t)$ &  - & $\hat{\mathcal{R}}_1(t)$ \\ \hline
		2 & $\hat{\mathcal{R}}_1(t),\hat{\mathcal{R}}_3(t),\hat{\mathcal{R}}_4(t)$ & - & $\hat{\mathcal{R}}_2(t)$ \\ \hline
		4 & $\hat{\mathcal{R}}_1(t),\hat{\mathcal{R}}_2(t),{\mathcal{R}}_3(t)$ & - & $\hat{\mathcal{R}}_4(t)$ \\ \hline  
	\end{tabular}
	\label{Param_Tube}
\end{table}

\begin{subequations}
\begin{align}
    \hat{\mathcal{R}}_1(t) & = \begin{cases} 
      {\mathcal{R}}_1(t), &\forall t \in [0, t_a), \\
      {\mathcal{R}}_1(t_a-\delta), &\forall t \in [t_a, t_b], \\
      {\mathcal{R}}'_1(t-t_b), &\forall t \in (t_b, t_p],
      \label{Parametric1}
   \end{cases} \\
    \hat{\mathcal{R}}_2(t) & = \begin{cases} 
      {\mathcal{R}}_2(t), &\forall t \in [0, t_c), \\
      {\mathcal{R}}_2(t_c-\delta), &\forall t \in [t_c, t_d], \\
      {\mathcal{R}}'_2(t-t_d), &\forall t \in (t_d, t_p], 
      \label{Parametric2}
   \end{cases} \\
    \hat{\mathcal{R}}_4(t) & = \begin{cases} 
      {\mathcal{R}}_4(t), &\forall t \in [0, t_e), \\
      {\mathcal{R}}_4(t_e-\delta), &\forall t \in [t_e, t_f], \\
      {\mathcal{R}}'_4(t-t_f), &\forall t \in (t_f, t_p]. 
   \end{cases}
\label{Parametric3}
\end{align}
\end{subequations}

\section{Controller synthesis}
\label{CtrlSyn}
Thus, the RASCA specification can be enforced by constraining the state trajectory of each agent $i$ within the respective negotiated and obstacle-circumvented tubes $\tilde{\mathcal{R}}_i^{(k)}(t) = [\tilde{\gamma}^{(k)}_{i,L}(t),\tilde{\gamma}^{(k)}_{i,U}(t)]$ as follows: 
\begin{gather}
    \tilde{\gamma}^{(k)}_{i,L}(t) < x^{(k)}_i(t) < \tilde{\gamma}^{(k)}_{i,U}(t), \forall (t,i,k) \in [0,t_p^{(i)}] \times \mathcal{I} \times [1;n_i]. \label{eqn:ppc}
\end{gather}

\begin{theorem}
    Given, negotiated and obstacle-circumvented tubes $\tilde{\mathcal{R}}_i^{(k)}(t) = [\tilde{\gamma}^{(k)}_{i,L}(t),\tilde{\gamma}^{(k)}_{i,U}(t)]$ for agent $i \in \mathcal{I}$ and dimension $k \in [1;n_i]$, the control law 
    \begin{equation}
    \label{Ctrlr}
    	u_i^{(k)}(x_i^{(k)},t) = -\kappa_i^{(k)}  \xi_i^{(k)} (x_i^{(k)},t) \varepsilon_i^{(k)}(x_i^{(k)},t), \kappa_i^{(k)} \in \mathbb{R}^+,
    \end{equation}
    confines the state trajectory of the agent to its corresponding tube i.e $x^{(k)}_i(t) \in \tilde{\mathcal{R}}_i^{(k)}(t)$ for all $t \in \mathbb{R}^+$.
    
    Here,
    $$\varepsilon_i^{(k)}(x_i^{(k)},t) = \ln \left( \frac{1+e_i^{(k)}(x_i^{(k)},t)}{1-e_i^{(k)}(x_i^{(k)},t)} \right)$$
    $$e_i^{(k)}(x_i^{(k)},t) = \frac{x_i^{(k)} - \frac{1}{2}(\tilde{\gamma}_{i,U}^{(k)}(x_i^{(k)},t) + \tilde{\gamma}_{i,L}^{(k)}(x_i^{(k)},t))}{\frac{1}{2} (\tilde{\gamma}_{i,U}^{(k)}(x_i^{(k)},t) - \tilde{\gamma}_{i,L}^{(k)}(x_i^{(k)},t))},$$
    $$\xi_i^{(k)}(x_i^{(k)},t) = \frac{4(1 - (e_i^{(k)}(x_i^{(k)},t))^2)^{-1}}{(\tilde{\gamma}_{i,U}^{(k)}(x_i^{(k)},t) - \tilde{\gamma}_{i,L}^{(k)}(x_i^{(k)},t))}.$$
\end{theorem}
Due to space constraints, we are omitting the proof in this paper. The details on the derivation of the controller can be found in \cite{das2024prescribed}.

\begin{remark}
    The feedback controller in Equation \eqref{Ctrlr} is entirely independent of system dynamics and disturbances while ensuring the system state remains within the prescribed spatiotemporal tubes (STTs), guaranteeing robustness and constraint satisfaction. 
\end{remark}

\section{Simulations}
\label{Sim}
\subsection{Case Study-1: Robot Navigation}
We demonstrate the effectiveness of the RASCA task algorithm on a homogeneous multi-agent system, where each agent is a three-wheeled omni-directional robot whose dynamics are defined as:

\begin{figure*}[h]
  \includegraphics[width=\textwidth]{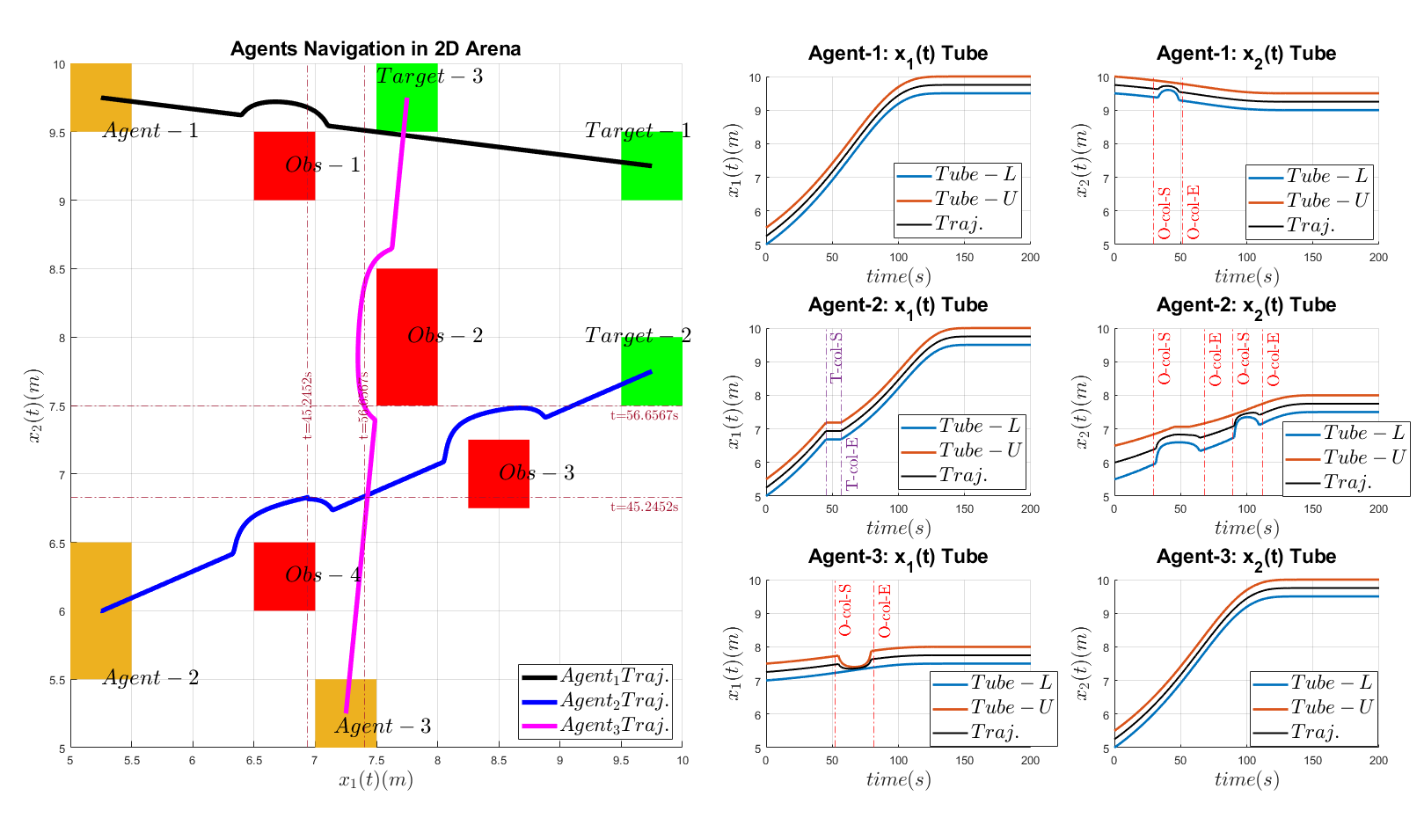}
  \caption{RASCA task in a 2D arena for Robotic navigation and respective temporal tubes (T-Col-S and T-Col-E are tube collision $\textit{Start}$ and $\textit{End}$ instances and O-Col-S and O-Col-E are obstacle collision $\textit{Start}$ and $\textit{End}$ instances respectively)}
  \label{2Darena}
\end{figure*}


\begin{equation}
	\begin{aligned}
		\left[\begin{array}{c}
			\dot{x}_1 \\
			\dot{x}_2 \\
			\dot{x}_3
		\end{array} \right]
		= 
		\left[ \begin{array}{ccc}
			\cos x_3 & -\sin x_3 & 0 \\
			\sin x_3  & \cos x_3 & 0 \\
			0 & 0 & 1 \\
		\end{array} \right]  
		\left[ \begin{array}{c}
		v_1 \\
		v_2 \\
		\omega \\
		\end{array} \right] + w(t)  
	\end{aligned}
	\label{OmniDyn}
\end{equation}
where the states of the robot are defined by the vector $[x_1,x_2,x_3]^T$ and  inputs as $[v_1,v_2,\omega]$, which is the vector of velocities in $x_1$ and $x_2$ directions, $\omega$ is the wheel angular velocity and $w$ is the bounded external disturbance. The arena is a two-dimensional space with limits in $x_1$ and $x_2$ directions being $[0,10]m$. Let the number of agents in the arena be $n=3$ with their starting positions belonging to sets represented as hyper-rectangles $Ag_1 = [5,5.5;9.5,10]$, $Ag_2 = [5,5.5;5.5,6.5]$ and $Ag_3 = [7,7.5;5,5.5]$. The arena has four static obstacles at hyper-rectangular locations $O_1 = [6.5,7;9,9.5], O_2 = [7.5,8;7.5,8.5], O_3 = [8.25,8.75;4.75,5.25]$ and $O_4 = [6.5,7;6,6.5]$. The target location hyper-rectangles for the agents are $T_1 = [9.5,10;9,9.5]$, $T_2 = [9.5,10;7.5,8]$ and $T_3 = [7.5,8;9.5,10]$ and the prescribed time for all the agents is $t_p=200s$. Figure \ref{2Darena} shows the spatiotemporal tubes along with tube-obstacle collision instances. As shown in Figure \ref{simul}, agents $2$ and $3$ collide during time instances $t = 45.2452 s$ and $t = 56.6567 s$, which is averted owing to negotiation-based tube parameterization. An animation of this simulation can be found in \cite{Video2024}. 
\begin{figure}[htpb]
	\centering
	\includegraphics[width=0.5\linewidth]{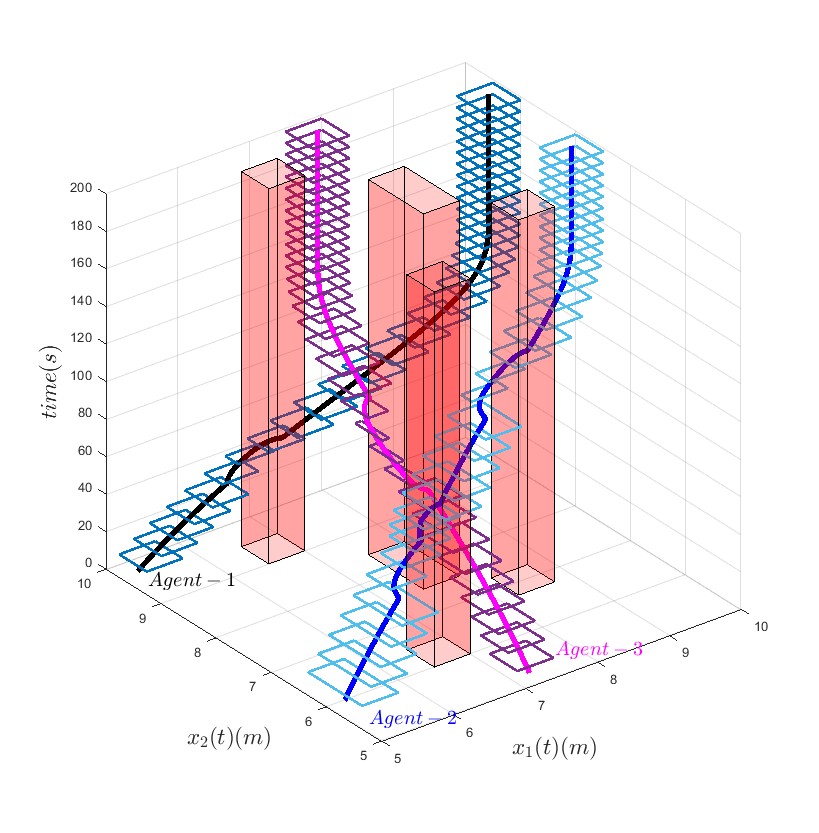}
	\caption{3 Agents navigation: RASCA task }
	\label{simul}
\end{figure}

\subsection{Case Study-2: Drone Navigation}
In this case study, we demonstrate a simulation collision avoidance scenario of randomly initialized initial and target locations of multi-agent system of drones with agent dynamics adapted from \cite{APF}: 
\begin{equation}
\label{Drone_DiD}
\begin{aligned}
 \dot{p}_{i} = q_{i}  
\end{aligned}
\end{equation}
where, $p_i$ capture the position and $q_i$ is the velocity input of the $i^{th}$ drone. The initial and target locations were three-dimensional hyperpolygons, that is cuboids. The initial location cuboids of all the agents are given as,
$Ag_1 = [2.64,3.14;3.5,4;3.27,3.77]'$, 
$Ag_2 = [2.93,3.43;2.99,3.5;4.22,4.72]'$, 
$Ag_3 = [2.4,2.9;1.79,2.29;3.02,3.52]'$,
$Ag_4 = [1.92,2.48;0.59,1.09;1.97,2.47]'$,
$Ag_5 = [2.46,2.96;1.78,2.28;1.8,2.29]'$ and 
$Ag_6 = [3.38,3.88;2.35,2.85;2.20,2.70]'$. 
The three-dimensional arena limits are $[0,5]m$ along all the dimensions. In the navigation task, tubes of agents $1$ and $3$ have collisions between time instances $t = 32.03s$ and $t = 51.85s$, and the tubes of agents $5$ and $6$ have collisions between time instances $t = 43.043 s$ and $t = 43.243 s$. The three-dimensional tubes are temporally parameterized through negotiation such that collision-free tubes are obtained as shown in the drone navigation scenario in Figure \ref{Drone}.

\begin{figure}[htpb]
	\centering
	\includegraphics[width=0.5\linewidth]{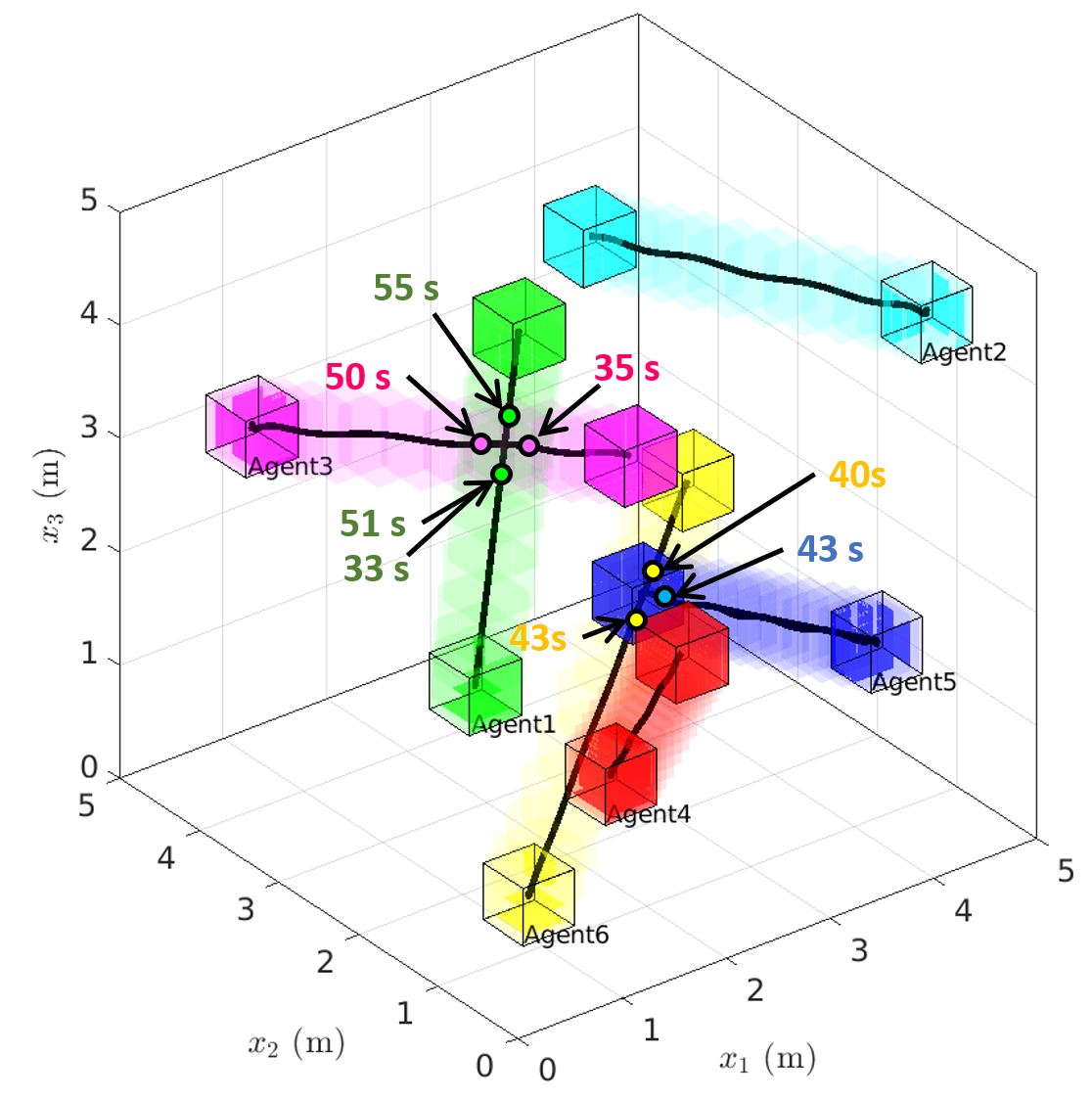}
	\caption{Drone Navigation: Collision Avoidance}
	\label{Drone}
\end{figure}

\subsection{Comparison}
The framework ensures zero collisions, even in compact environments, by temporally halting agents when necessary. In contrast to standard multi-agent path planning techniques, which struggle with collision avoidance in dense arenas, and symbolic control methods, which are computationally expensive, our approach provides an efficient and scalable solution. Table \ref{tab:comparison} below summarizes the key advantages of our proposed framework compared to existing techniques:

\begin{table*}[h!]
\centering
\caption{Comparison of the proposed spatiotemporal tubes-based framework with standard multi-agent path planning and symbolic control techniques}

\begin{tabular}{|p{3cm}|p{4cm}|p{4cm}|p{4.5cm}|}
\hline
\textbf{Feature} & \textbf{Proposed Framework} & \textbf{Multi-Agent Path Planning \cite{lin2022review}} & \textbf{Symbolic Control Techniques \cite{tabuada2009verification}} \\ 
\hline
\textbf{Collision Avoidance in Compact Arena} & Zero collisions, even in compact environments by temporally halting agents. & Collision-free paths but less effective in very compact arenas. & Achieves collision avoidance but computationally expensive. \\ 
\hline
\textbf{Handling Prescribed-Time Tasks} & Handles prescribed-time reach-avoid-stay tasks efficiently. & Does not inherently handle prescribed-time specifications. & Handles temporal constraints but at high computational cost. \\ 
\hline
\textbf{Computational Efficiency} & Highly efficient with a closed-form and approximation-free solution. & Efficient but may struggle in dense environments. & Computationally intensive due to state-space abstractions. \\ 
\hline
\textbf{Scalability} & Scales effectively with multiple agents. & May face issues scaling in complex, crowded arenas. & Limited scalability due to high computational requirements. \\ 
\hline
\end{tabular}
\label{tab:comparison}
\end{table*}

\section{Conclusion and Future Work}
This study explained path planning and navigation using a negotiation framework for a multi-agent system specified in the RASCA task. This was achieved using the STT and an associated controller that guarantees collision free RAS task in the prescribed time. The negotiation process continues iteratively until all the agents obtain parameterized tubes that cover the entire prescribed time range. The algorithm proposed in this study was verified using a path planning and robot navigation task on an omni-directional robot in a two-dimensional arena and a drone navigation task in a three-dimensional arena. 
It can also be observed that the path planning achieved using negotiations is distributed, in the sense that once negotiation terminates, agents don't depend on other agents or a centralized entity for collision free navigation. 
    
In this study we consider a network topology that sequentially selects a hub agent. In future work we plan to introduce a  game-theoretic framework that assigns priority to agents in a compact and complex arena which has tighter prescription of reach time on agents. Scenarios like overlap of target space of one agent with STT of other agent can create stalemate situations with sequential order of negotiation, this can be mitigated with an agent-wise priority structure. The priority structure enforces an order in the negotiation sequence, so that the right-of-way in a navigation task is allotted to an agent in need accordingly, such that for all the agents specified RASCA task is completed in the prescribed time, despite heterogeneity in agent dynamics.

\bibliographystyle{unsrt}
\bibliography{./MyLibrar/MyLibrary}

\end{document}